\documentclass[preprint,12pt,authoryear]{elsarticle}

\usepackage{amssymb}
\usepackage{amsmath}
\usepackage{amsthm}

\usepackage{amsfonts}

\usepackage{mathtools}

\usepackage{hyperref}

\usepackage{float}

\usepackage[round,authoryear]{natbib}

\newtheorem*{theorem*}{Theorem}
\newtheorem{theorem}{Theorem}[section]

\newtheorem{proposition}[theorem]{Proposition}

\newtheorem{lemma}[theorem]{Lemma}

\newtheorem{definition}[theorem]{Definition}

\let\temp\varphi
\let\varphi\phi
\let\phi\temp

\newcommand{\N}{\mathbb{N}}
\newcommand{\R}{\mathbb{R}}
\newcommand{\Z}{\mathbb{Z}}
\newcommand{\Q}{\mathbb{Q}}

\DeclarePairedDelimiter{\iop}{(}{)}
\DeclarePairedDelimiter{\iob}{[}{]}
\DeclarePairedDelimiter{\ios}{\{}{\}}
\DeclarePairedDelimiter{\absolute}{|}{|}

\newcommand{\p}{\iop*}
\newcommand{\s}{\ios*}
\renewcommand{\b}{\iob*}
\newcommand{\abs}{\absolute*}

\DeclareMathOperator*{\supp}{supp}

\DeclareMathOperator{\NE}{NE}
\DeclareMathOperator{\PNE}{PNE}

\DeclareMathOperator*{\remn}{Remainder}

\journal{Games and Economic Behavior}

\begin{document}

\begin{frontmatter}

%% Title, authors and addresses

%% use the tnoteref command within \title for footnotes;
%% use the tnotetext command for theassociated footnote;
%% use the fnref command within \author or \affiliation for footnotes;
%% use the fntext command for theassociated footnote;
%% use the corref command within \author for corresponding author footnotes;
%% use the cortext command for theassociated footnote;
%% use the ead command for the email address,
%% and the form \ead[url] for the home page:
%% \title{Title\tnoteref{label1}}
%% \tnotetext[label1]{}
%% \author{Name\corref{cor1}\fnref{label2}}
%% \ead{email address}
%% \ead[url]{home page}
%% \fntext[label2]{}
%% \cortext[cor1]{}
%% \affiliation{organization={},
%%            addressline={}, 
%%            city={},
%%            postcode={}, 
%%            state={},
%%            country={}}
%% \fntext[label3]{}

%% \title{Title\tnoteref{label1}}
%% \tnotetext[label1]{}
%% \author{Name\corref{cor1}\fnref{label2}}
%% \ead{email address}
%% \ead[url]{home page}
%% \fntext[label2]{}
%% \cortext[cor1]{}
%% \affiliation{organization={},
%%            addressline={}, 
%%            city={},
%%            postcode={}, 
%%            state={},
%%            country={}}
%% \fntext[label3]{}

\title{Nash Equilibria with Irradical Probabilities\tnoteref{t1}} %% Article title
\tnotetext[t1]{Funding: This work was supported by DARPA E-BOSS HR001124C0486 and Hasso Plattner Institute.}

%% use optional labels to link authors explicitly to addresses:
%% \author[label1,label2]{}
%% \affiliation[label1]{organization={},
%%             addressline={},
%%             city={},
%%             postcode={},
%%             state={},
%%             country={}}
%%
%% \affiliation[label2]{organization={},
%%             addressline={},
%%             city={},
%%             postcode={},
%%             state={},
%%             country={}}

%% \author{Name\corref{cor1}\fnref{label2}}
%% \ead{email address}
%% \ead[url]{home page}
%% \fntext[label2]{}
%% \cortext[cor1]{}
%% \affiliation{organization={},
%%             addressline={},
%%             city={},
%%             postcode={},
%%             state={},
%%             country={}}

\author{Edan Orzech\corref{cor1}}\ead{iorzech@csail.mit.edu}
\author{Martin Rinard}\ead{rinard@csail.mit.edu}
%% Author name

\cortext[cor1]{Corresponding author}

%% Author affiliation
\affiliation{organization={MIT CSAIL},%Department and Organization
            addressline={32 Vassar Street}, 
            city={Cambridge},
            postcode={02139}, 
            state={Massachusetts},
            country={USA}}

%% Abstract
\begin{abstract}
%% Text of abstract
    We present for every $n\ge4$ an $n$-player game in normal form with payoffs in $\{ 0, 1, 2\}$ that has a unique, fully mixed, Nash equilibrium in which all the probability weights are irradical (i.e., algebraic but not closed form expressible even with $m$-th roots for any integer $m$).
\end{abstract}

%% Keywords
\begin{keyword}
%% keywords here, in the form: keyword \sep keyword
    Nash equilibrium \sep algebraic \sep irrational
%% PACS codes here, in the form: \PACS code \sep code

%% MSC codes here, in the form: \MSC code \sep code
%% or \MSC[2008] code \sep code (2000 is the default)
    \MSC[2020] 91A06 \sep 91A10
\end{keyword}

\end{frontmatter}

\section{Introduction}
Every Nash equilibrium (NE) in a 2-player normal form game with rational payoffs is a solution to a 
linear program with rational weights \citep{roughgarden2010algorithmic}. Therefore, if 
$\sigma$ is a NE then there is a NE $\tau$ with rational probability weights (henceforth we 
refer to them as \emph{probabilities}) and support $\supp\sigma$. However, this conclusion 
fails to hold in $n$-player games when $n\ge3$. This is because more generally, NEs are 
solutions of polynomial equations and inequalities in $n$ variables and degree at most $n-1$,
and these might not have solutions in the rationals. A concrete example is found in 
\citet{bilo2014complexity}: a $3$-player game with integer payoffs in $\s{0,1,2,3}$, 
possessing a unique NE, and its probabilities are $((x_1,1-x_1),(x_2,1-x_2),(x_3,1-x_3))$ 
where
$x_1=\frac{25-\sqrt{409}}{12},~x_2=\frac{13+\sqrt{409}}{60},~x_3=\frac{21-\sqrt{409}}{2}$.  
These 6 probabilities are all irrational. A different known example is the 3-player poker game in 
\citet{nash1950simple}.

Such a result is possible for every $n\ge4$ as well but there seems to be no previously known game with this property for any such $n$.
In this paper we present for every $n\ge4$ an $n$-player game $G_n$ that has 2 actions per 
player and the following properties: (1) The payoffs are in $\s{0,1,2}$; (2) $G_n$ has a 
unique NE; (3) that NE is fully mixed, and (4) the probability weights of this NE are all 
inexpressible with radicals (henceforth we call such numbers \emph{irradicals}).  
Inexpressibility with radicals is an even stronger requirement than irrationality -- it 
means that these probabilities do not have closed form expressions using rational arithmetic and $m$-th roots for any integer $m$.

To play a mixed NE, a player must randomly select an action from the underlying (discrete) probability distribution.
Exact sampling from discrete distributions with rational probabilities is a well understood problem solved with efficient algorithms~\citep{knuth1976complexity,draper2025efficient}. There, the computations are done using integer arithmetic which be done precisely on a modern digital computer. But sampling from distributions with irrational (or irradical) probabilities using rational arithmetic is a less studied problem. Irrational numbers can only be approximated by finite integers. With current known techniques, exact sampling from a distribution with irrational weights can only be done with arbitrary precision arithmetic using a power series expansion of the probabilities~\citep{flajolet2011buffon,mendo2020simulating}, which is computationally and pragmatically more involved than sampling from distributions with rational probabilities.
In our context, all the probability weights are algebraic numbers, so the sampling problem appears to require computing roots of polynomials with potentially very large degrees and coefficients (the polynomials associated with $G_5$ have degree $26$ and $10$-digit coefficients).
This fact highlights a challenge that players face in this context.

For example, our $4$-player game has a unique NE, 
$((x_1,1-x_1),(x_2,1-x_2),(x_3,1-x_3),(x_4,1-x_4))$ such that $x_1,x_2,x_3,x_4$ are these 
roots of the polynomials
\begin{align*}
  P_1(z)&=7z^6-42z^5+89z^4-83z^3+40z^2-10z+1,\\
  P_2(z)&=4z^6-27z^5+70z^4-79z^3+45z^2-13z+1,\\
  P_3(z)&=140z^6-511z^5+701z^4-454z^3+141z^2-19z+1,\\
  P_4(z)&=5z^6-44z^5+143z^4-163z^3+85z^2-21z+2,
\end{align*}
that are approximately
$x_1\approx0.529,~x_2\approx0.846,~x_3\approx0.523,~x_4\approx0.320$. $x_1,x_2,x_3,x_4$ 
cannot be expressed with radicals because $P_1,P_2,P_3,P_4$ are all irreducible polynomials 
over the rationals whose Galois groups are all $S_6$ which is not solvable 
(see~\citet{artin2011algebra} for reference).

To obtain the main result we first construct the mentioned $4$-player game (Proposition~\ref{prop:n=4}) and 
a $5$-player game (Proposition~\ref{prop:n=5}) with a unique NE whose probabilities are 
irradical. We then use these two games to construct such a game for any $n\ge6$ 
(Proposition~\ref{prop:n>=6}). More specifically, for $n\ge6$ we juxtapose independent copies of 
the $4$-player and $5$-player games, and to one of the game copies possibly attach 2 
additional players who force each other to mimic the mixed strategy of one player in the game 
copy.

We complement this result by showing that every $2\times2\times2$ game with integer payoffs 
has a NE whose probabilities can be expressed using only rationals and square roots of 
rationals (Proposition~\ref{prop:n=3-sqrt}).
We also present a simple $3$-player game with integer payoffs in only $\s{0,1,2}$ whose 
unique NE has only irrational probabilities (Proposition~\ref{prop:n=3}). This game is 
simpler than that of \citet{bilo2014complexity}.
Its unique NE $((x_1,1-x_1),(x_2,1-x_2),(x_3,1-x_3))$ where
$x_1=\frac{7-\sqrt{13}}{6},~x_2=\frac{7+\sqrt{13}}{18},x_3=\frac{-3+\sqrt{13}}{2}$.
These are approximately $x_1\approx0.566,~x_2\approx0.589,~x_3\approx0.303$.

Note that given such a $3$-player game it is easy to extend it to an $n$-player game where 
\emph{some} of the players need to play mixed actions with irrational probabilities in 
every NE -- simply add players with one dominating action each. But then only the first 3 
players will need to play NE strategies with irrational probabilities. We are interested in 
forcing \emph{all} players to play NE strategies with irrational probabilities. For this 
reason we are interested only in fully mixed equilibria.

It is also possible to use such a $3$-player game to generate $n$-player games for every 
$n\ge5$ where all players need to play NE strategies with irrational probabilities. This is 
achieved by applying our results in Section~\ref{sec:ext} to such a $3$-player game.  
However, this approach alone leaves unresolved the case of 4 players, and all the probabilities in the other cases will not be irradical.

There are some related results in the literature. First is the characterization of 
\citet{kreps1981finite} of (mixed) strategy profiles which can be the unique NEs of a 
suitable normal form game. In particular, any fully mixed strategy profile for $n$ players 
with 2 actions each is the unique NE of some $n$-player game (with 2 actions per player).  
However, the payoffs in the constructed game will 
involve the probabilities in that strategy profile. In particular, if the probabilities are 
irrational then the payoffs found will be irrational. Here we are looking for games with 
integer payoffs, since these are simpler to store on a computer than irrational payoffs.

Another result is by \citet{bubelis1979equilibria} who constructs for every algebraic 
number $\alpha$ a $3$-player game with integer payoffs such that $\alpha$ is the payoff to 
some player in the NE of the game. However, in the constructed game player~$3$ has the fully mixed 
uniform distribution (which is comprised of rational probabilities) as an equilibrium 
strategy.

Another related result is the universality of NEs by \citet{datta2003universality}. This 
result asserts that every semialgebraic set is isomorphic to the set of NEs of some game.  
However, an isomorphism in this context only 
preserves the topological structure of the semialgebraic set rather than its exact members.  
In particular, in \citet{datta2003universality}, a semialgebraic set consisting of a single 
point with irrational entries will be mapped to a NE profile with rational probabilities.

\paragraph{Paper organization} In Section~\ref{sec:defs} we define the setup and provide 
algebra background. In Section~\ref{sec:3} we present our results for $3$-player, 
$2\times2\times2$ games. In Section~\ref{sec:4} and Section~\ref{sec:5} we present the 
$4$-player game and $5$-player game. In Section~\ref{sec:ext} we use the latter two games 
to construct the remaining $n$-player games for every $n\ge6$.

\section{Setup and lemmas}\label{sec:defs}

The games we define are in the usual normal form.
They will be presented as a 2-column table of rows in the format
\begin{center}
  [action tuple $|$ payoffs].
\end{center}
The action tuple and the payoffs in each row are ordered by the players $1,\ldots,n$.

In the games we present each player has 2 actions called $0$ and $1$. Each mixed action 
profile is defined by $x=(x_1,\ldots,x_n)$ where $x_i\in[0,1]$ is the probability player $i$ 
gives to playing action 0. The payoff function of player~$i$ is $u_i$. We then define for 
every $i$ the polynomial
\begin{align}\label{eq:polys}
  f_i(x_{-i})=u_i(x_1,\dots,x_{i-1},1,x_{i+1},\ldots,x_n)-u_i(x_1,\ldots,x_{i-1},0,x_{i+1},\ldots,x_n).
\end{align}
This is the expected payoff advantage action $0$ has over action $1$ for player $i$ given 
the mixed actions~$x_{-i}$ of the other players. We sometimes write $f_i(x)$ or $f_i$ for 
simplicity.

Given a game $G$ let $\NE(G)$ be the set of NEs of $G$ and $\PNE(G)$ the set of pure NEs of 
$G$. For reference we state the following lemma which can be proved by definition.

\begin{lemma}\label{lmm:ne-poly-cond}
  If $x\in\NE(G)$ then for every $i$,
  \begin{align}\label{eq:poly-cond}
    f_i(x_{-i})>0&~\Rightarrow~x_i=1,\nonumber\\
    f_i(x_{-i})<0&~\Rightarrow~x_i=0,\\
    0<x_i<1&~\Rightarrow~f_i(x_{-i})=0.\nonumber
  \end{align}
\end{lemma}

\subsection{Algebra background}
Here we cover the necessary algebra background for this paper (see 
\citet{artin2011algebra} for a detailed reference). Let $\Q[x]$ be the ring (set) of 
polynomials~in the variable $x$ with rational coefficients. $P\in\Q[x]$ is irreducible if it cannot be written as 
$P_1P_2$ for $P_1,P_2\in\Q[x]$ where $0<\deg P_1,\deg P_2<\deg P$.
A~number $a$ is called \emph{algebraic} if there is a $0\ne P\in\Q[x]$ such that $P(a)=0$. Every 
algebraic number has a \emph{minimal polynomial}: a unique, monic polynomial of minimum 
degree $P_a\in\Q[x]$ such that $P_a(a)=0$. For example, $P_{\sqrt2}=x^2-2$. Here we will not insist on monicity in our search of these polynomials.

One way to prove a number $a$ is irrational is by showing that $\deg P_a\ge2$. To verify 
that a polynomial at hand $P$ is $P_a$, it is sufficient to show that $P(a)=0$ and $P$ is irreducible in $\Q[x]$ (and that $P$ is monic).

An irrational algebraic number can often be represented in a compact way, using 
\emph{radicals} -- $m$-roots ($\sqrt[m]~$ for an integer $m$) of rationals. For example $\sqrt2$ represents the positive root of 
$x^2-2$. Another example is~$\sqrt2+3\sqrt[3]5$. Given a polynomial with a root $a$, we 
want to know when it is possible to construct such a representation for $a$. This 
will be relevant in Section~\ref{sec:4} and Section~\ref{sec:5}.

Let $P$ be an irreducible polynomial with degree $n$ and roots $a_1,\ldots,a_n$. The field 
$F=\Q[a_1,\ldots,a_n]$ is defined to be the smallest field that contains $\Q$ and 
$a_1,\ldots,a_n$.\footnote{$\Q[a_1,\ldots,a_n]$ is a field while $\Q[x]$ and $\Q[x_1,\ldots,x_n]$ 
are rings because $x_1,\ldots,x_n$ are symbolic variables with no defined inverses, while 
$a_1,\ldots,a_n$ are real numbers that do have their inverses $\frac{1}{a_1},\ldots,\frac{1}{a_n}$ in $\Q[a_1,\ldots,a_n]$.} For example, 
$\Q\b{\sqrt2}$ is the set of all numbers $a+b\sqrt2$ where $a,b\in\Q$. An 
\emph{automorphism of $F$ that extends $\Q$} is a function $f:F\to F$ that is a bijection, 
preserves the field operations (e.g., $f(a+b)=f(a)+f(b)$) and maps every rational to 
itself. It can be shown that $f$ is uniquely defined by its action on the roots 
$a_1,\ldots,a_n$. Furthermore, $f\p{\s{a_1,\ldots,a_n}}=\s{a_1,\ldots,a_n}$, since for every $i$ it 
holds that $P(f(a_i))=f(P(a_i))=f(0)=0$, meaning $f(a_i)$ is a root of $P$. The set $Aut(F)$ 
is called the \emph{Galois group} of $P$.

A fundamental result in algebra states that a number $a$ can be represented with radicals 
if and only if the Galois group of its minimal polynomial $P_a$ is \emph{solvable}. The 
formal definition of solvability does not matter here. The point is that to determine if 
$a$ can be written with radicals we only need to determine if the Galois group of $P_a$ is 
solvable (and wlog we can ignore monicity) A particular type of group that is \emph{unsolvable} is $S_n$ for $n\ge5$. $S_n$ 
is the set of all $n!$ permutations of $\s{1,\ldots,n}$ with the group operation of 
permutation composition.

Take two examples. The first is $\sqrt2$. Its minimal polynomial is $x^2-2$. Its Galois 
group is called~$\Z_2$, which is simply the set $\s{0,1}$ equipped with the XOR operation.  
This group is solvable, and indeed $\sqrt2$ is a radical representation of the positive 
root of $x^2-2$. The second example is any root of $P=x^5-x-1$. It can be shown that the 
Galois group of $P$ is $S_5$. It is unsolvable, and indeed the roots of $P$ cannot be 
represented with radicals. We call such numbers \emph{irradicals}.

Note that the pair $(x_i,1-x_i)$ satisfies $x_i\in\Q\iff1-x_i\in\Q$, and similarly for 
expressbility with radicals. Therefore in our proofs of irrationality and irradicality we 
only analyze $x_1,\ldots,x_n$.

\section{Three players}\label{sec:3}

Define the $2\times2\times2$ game $G_3$:
\begin{figure}[H]
\centering
\begin{align*}
  \left[\begin{array}{ccc|ccc}
    0&0&0&0&0&2\\
    0&0&1&1&0&0\\
    0&1&0&0&1&0\\
    0&1&1&0&0&1\\
    1&0&0&0&0&0\\
    1&0&1&0&1&1\\
    1&1&0&1&1&0\\
    1&1&1&1&0&1
  \end{array}\right]
\end{align*}
\caption{The game $G_3$}
\label{fig:G3}
\end{figure}

\begin{proposition}\label{prop:n=3}
  $G_3$ has a unique NE $((x_1,1-x_1),(x_2,1-x_2),(x_3,1-x_3))$ given by
  \begin{align*}
    x_1=\frac{7-\sqrt{13}}{6},~x_2=\frac{7+\sqrt{13}}{18},~x_3=\frac{-3+\sqrt{13}}{2}.
  \end{align*}
  The probabilities in the NE are all irrational.
\end{proposition}

\begin{proof}
  First observe that $\PNE(G_3)=\emptyset$ using this table of profitable deviations:
  \begin{align*}
    \begin{array}{|c|c|}
      \hline
      \text{action profile}&\text{unsatisfied players}\\
      \hline
      (0,0,0)&2\\
      (0,0,1)&3\\
      (0,1,0)&1,3\\
      (0,1,1)&1\\
      (1,0,0)&2,3\\
      (1,0,1)&1\\
      (1,1,0)&3\\
      (1,1,1)&2\\
      \hline
    \end{array}
  \end{align*}
  
  Now the polynomials $f_1,f_2,f_3$ from Equation~\ref{eq:polys} are
  \begin{align*}
    f_1(x_2,x_3)&=-x_2x_3+2x_2-1,~f_2(x_1,x_3)=x_1x_3-x_1-2x_3+1,\\
    f_3(x_1,x_2)&=3x_1x_2-1.
  \end{align*}

  We show using Lemma \ref{lmm:ne-poly-cond} that if $x_i\in\s{0,1}$ for some $i$ then 
  $x\notin\NE(G_3)$. There are 6 cases.

  \paragraph{Case 1: $x_1=0$} then $f_2=-2x_3+1,~f_3=-1<0$, so $x_3=0$ so $f_2=1>0$ so 
  $x_2=1$. Therefore $x=(0,1,0)\notin\NE(G_3)$.
  
  \paragraph{Case 2: $x_1=1$} then $f_2=-x_3,~f_3=3x_2-1$.
  If $x_3>0$ then $f_2<0$ so $x_2=0$ so $f_3=-1<0$ so $x_3=0$ -- a contradiction. Therefore 
  $x_3=0$, so $f_3\le0$ so $x_2\le\frac{1}{3}$ so $f_1\le-\frac{1}{3}<0$ so $x_1=0$ -- a 
  contradiction.

  \paragraph{Case 3: $x_2=0$} then $f_1=-1$ so $x_1=0$ so by case 1 $x\notin\NE(G_3)$.

  \paragraph{Case 4: $x_2=1$} then $f_1=-x_3+1,~f_3=3x_1-1$. By cases 1 and 2, 
  $x_1\in(0,1)$ so $f_1=0$ so $x_3=1$ so $f_2=-1<0$ so $x_2=0$ -- a contradiction.  
  Therefore $x\notin\NE(G_3)$.

  \paragraph{Case 5: $x_3=0$} then $f_2=-x_1+1$. By cases 3 and 4, $x_2\in(0,1)$ so 
  $f_2=0$ so $-x_1+1=0$ so $x_1=1$ so by case 2 we conclude that $x\notin\NE(G_3)$.

  \paragraph{Case 6: $x_3=1$} then $f_1=x_2-1$. Again $0=f_1=x_2-1$ so $x_2=1$, so by case 
  4 we conclude that $x\notin\NE(G_3)$.

  \bigskip

  Therefore the only NEs are fully mixed. This means that $x\in\NE(G_3)$ satisfies 
  $f_i(x)=0$ for every $i$. By $f_3$ we have $x_1=\frac{1}{3x_2}$.
  Then by $f_2$ we have $\frac{x_3}{3x_2}-\frac{1}{3x_2}-2x_3+1=0$ so $(6x_3-3)x_2=x_3-1$.  
  $x_3\ne\frac{1}{2}$ because $x_3\ne1$ so $x_2=\frac{x_3-1}{6x_3-3}$.
  Then by $f_1$ we have $\frac{(x_3-1)(2-x_3)}{6x_3-3}=1$ so $-x_3^2+3x_3-2=6x_3-3$ so 
  $x_3^2+3x_3-1=0$ so because $x_3>0$ we get $x_3=\frac{-3+\sqrt{13}}{2}$, 
  $x_2=\frac{-5+\sqrt{13}}{-24+6\sqrt{13}}=\frac{7+\sqrt{13}}{18}$ and 
  $x_1=\frac{6}{7+\sqrt{13}}=\frac{7-\sqrt{13}}{6}$.

  $x_1,x_2,x_3\notin\Q$ because they are expressions involving only rationals and 
  $\sqrt{13}$ which is irrational. Their minimal polynomials in $\Q[y]$ are
  \begin{align*}
    P_1(y)=3y^2-7y+3,~P_2(y)=9y^2-7y+1,~P_3(y)=y^2+3y-1.
  \end{align*}
  These are all indeed irreducible by the rational root theorem.
\end{proof}

We conclude with a proof that square roots of rationals (along with rational arithmetic) 
are always enough to represent a NE of a $2\times2\times2$ game.

\begin{proposition}\label{prop:n=3-sqrt}
  Let $G$ be a $2\times2\times2$ game with integer payoffs. There is a NE 
  $((x_1,1-x_1),(x_2,1-x_2),(x_3,1-x_3))$ such that for every $i\in\s{1,2,3}$, 
  $x_i=a_i+b_i\sqrt c_i$ for some $a_i,b_i,c_i\in\Q$.
\end{proposition}
\begin{proof}
  In full generality, the polynomials $f_1,f_2,f_3$ have the form
  \begin{align*}
    f_1(x_2,x_3)&=A_1x_2x_3+B_1x_2+C_1x_3+D_1,\\
    f_2(x_1,x_3)&=A_2x_1x_3+B_2x_1+C_2x_3+D_2,\\
    f_3(x_1,x_2)&=A_3x_1x_2+B_3x_1+C_3x_2+D_3,
  \end{align*}
  where the coefficients $A_i,B_i,C_i,D_i$ are all integers, as the payoffs in the game are 
  integers.

  Let $x\in\NE$.
  There are 3 cases up to permuting actions and payoffs:

  \paragraph{Case 1: $x_1=0$ and $x_2,x_3\in(0,1)$} Then 
  $f_1\le0,~f_2=C_2x_3+D_2=0,~f_3=C_3x_2+D_3=0$.

  From $f_1\le0$ we get $(A_1x_3+B_1)x_2\le-C_1x_3-D_1$. If $A_1x_3+B_1=0$ then either 
  $x_3=-\frac{B_1}{A_1}\in\Q$ or $A_1=B_1=0$, in which case $x_2,x_3$ are defined by linear 
  equations and inequalities with rational coefficients, so $x_2,x_3$ can be 
  rational in a NE.

  Otherwise, $x_2\le-\frac{C_1x_3+D_1}{A_1x_3+B_1}$.
  
  If $C_i\ne0$ for $i\in\s{2,3}$ then $x_{5-i}=-\frac{D_i}{C_i}\in\Q$. If $i=2$ then 
  $f_1=\p{-\frac{A_1D_2}{C_2}+B_1}x_2-\frac{C_1D_2}{C_2}+D_1\le0$. Since this inequality 
  holds for some value of $x_2$, the inequality also holds for some rational value of 
  $x_2$. In total, $x_1,x_2,x_3$ can all be rational (in a NE).

  Otherwise $C_2=C_3=0$. Then the only constraints on $x_2,x_3$ are being in $(0,1)$ and 
  the inequality $x_2\le-\frac{C_1x_3+D_1}{A_1x_3+B_1}$. This inequality is in particular 
  satisfied for some rational values of $x_2,x_3$, meaning that $x_1,x_2,x_3$ can all be 
  rational.
  
  \paragraph{Case 2: $x_1=x_2=0$ and $x_3\in(0,1)$} Then 
  $f_1=C_1x_3+D_1\le0,~f_2=C_2x_3+D_2\le0,~f_3=D_3=0$. So $x_3$ can be rational in a NE with $x_1=x_2=0$.

  \paragraph{Case 3: $x_1,x_2,x_3\in(0,1)$} Then $f_1=f_2=f_3=0$.
  Then $(A_1x_3+B_1)x_2=-C_1x_3-D_1$. Suppose that $(*)$ $A_1x_3+B_1\ne0$. Then 
  $x_2=-\frac{C_1x_3+D_1}{A_1x_3+B_1}$. From $f_3=0$ we then get 
  $\p{-\frac{A_3(C_1x_3+D_1)}{A_1x_3+B_1}+B_3}x_1-\frac{C_3(C_1x_3+D_1)}{A_1x_3+B_1}+D_3=0$.  
  So
  \begin{align*}
    \p{(B_3A_1-A_3C_1)x_3+B_3B_1-A_3D_1}x_1=(C_3C_1-D_3A_1)x_3+C_3D_1-D_3B_1.
  \end{align*}
  Suppose that $(**)$ $(B_3A_1-A_3C_1)x_3+B_3B_1-A_3D_1\ne0$. Then
  \begin{align*}
    x_1=\frac{(C_3C_1-D_3A_1)x_3+C_3D_1-D_3B_1}{(B_3A_1-A_3C_1)x_3+B_3B_1-A_3D_1}.
  \end{align*}
  
  Now $f_2=0$ simplifies into a quadratic equation in $x_3$ involving rational 
  coefficients.
  So $x_3$ can be written with rational arithmetic and (non-nested) square roots. The same then holds for $x_1,x_2$ by their expressions as ratios of linear 
  functions of $x_3$.

  Now suppose that $(*)$ and $(**)$ are not both true. There are 2 cases:
  \begin{itemize}
    \item $(*)$ is false, so $A_1x_3+B_1=0$. Then $C_1x_3+D_1=0$. If 
      $A_1=C_1=0$ then $B_1=D_1=0$ and $f_1$ is always $0$. If $f_2$ is always~$0$ then 
      $x_3$ can take any value in $(0,1)$ (particularly rational ones), and one can verify 
      that $x_1,x_2$ can take rational values. So $f_2,f_3$ are not always $0$, and because 
      of symmetry we can swap $f_1$ and $f_2$ to have one of $A_1\ne0$ and $C_1\ne0$ being 
      true which implies that $x_3\in\Q$. It can then be verified that $x_1,x_2$ can both 
      be rational.

    \item $f_1,f_3$ are not always $0$ and $A_1x_3+B_1\ne0$, $(*)$ is true and $(**)$ is false, so $(B_3A_1-A_3C_1)x_3+B_3B_1-A_3D_1=0$. Then 
      $(C_3C_1-D_3A_1)x_3+C_3D_1-D_3B_1=0$. If $B_3A_1-A_3C_1=C_3C_1-D_3A_1=0$ then 
      $C_3D_1-D_3B_1=B_3B_1-A_3D_1=0$ and conditioned on $f_1(x_2,x_3)=0$, we get that 
      $f_3$ is always $0$ in contradiction to our assumption.

      So one of $B_3A_1-A_3C_1$ and $C_3C_1-D_3A_1$ is not $0$. Then $x_3\in\Q$. It can 
      then be verified that $x_1,x_2$ can both be rational. \qedhere
  \end{itemize}
\end{proof}

\section{Four players}\label{sec:4}

Define the $2\times2\times2\times2$ game $G_4$:
\begin{figure}[H]
\centering
\begin{align*}
  \left[\begin{array}{cccc|cccc}
    0&0&0&0&0&0&0&1\\
    0&0&0&1&1&0&0&0\\
    0&0&1&0&0&0&1&0\\
    0&0&1&1&0&1&0&0\\
    0&1&0&0&0&0&0&2\\
    0&1&0&1&1&1&0&0\\
    0&1&1&0&1&0&1&0\\
    0&1&1&1&1&0&0&1\\
    1&0&0&0&0&0&0&0\\
    1&0&0&1&0&1&1&1\\
    1&0&1&0&1&0&1&1\\
    1&0&1&1&1&1&0&1\\
    1&1&0&0&1&1&0&0\\
    1&1&0&1&0&0&1&1\\
    1&1&1&0&0&1&1&0\\
    1&1&1&1&0&1&0&1
  \end{array}\right]
\end{align*}
\caption{The game $G_4$}
\label{fig:G4}
\end{figure}

\begin{proposition}\label{prop:n=4}
  $G_4$ has a unique NE $((x_1,1-x_1),(x_2,1-x_2),(x_3,1-x_3),(x_4,1-x_4))$ given by
  \begin{align*}
    x_1\approx0.529270752820,&~x_2\approx0.846414728986,\\
    x_3\approx0.523440476515,&~x_4\approx0.320065197645.
  \end{align*}
  The probabilities in the NE are all irrational and irradical.
\end{proposition}
\begin{proof}
  First observe that $\PNE(G_4)=\emptyset$ using this table of profitable deviations:
  \begin{align*}
    \begin{array}{|c|c|}
      \hline
      \text{action profile}&\text{unsatisfied players}\\
      \hline
      (0,0,0,0)&3\\
      (0,0,0,1)&2\\
      (0,0,1,0)&1\\
      (0,0,1,1)&1\\
      (0,1,0,0)&1,3\\
      (0,1,0,1)&4\\
      (0,1,1,0)&4\\
      (0,1,1,1)&2\\
      (1,0,0,0)&2,3,4\\
      (1,0,0,1)&1\\
      (1,0,1,0)&2\\
      (1,0,1,1)&3\\
      (1,1,0,0)&3,4\\
      (1,1,0,1)&1,2\\
      (1,1,1,0)&1,4\\
      (1,1,1,1)&1,3\\
      \hline
    \end{array}
  \end{align*}
  
  Now the polynomials $f_1,f_2,f_3,f_4$ from Equation~\ref{eq:polys} are
  \begin{align*}
    f_1(x_2,x_3,x_4)&=x_2x_3x_4+2x_2x_3-2x_2-2x_3x_4+1,\\
    f_2(x_1,x_3,x_4)&=3x_1x_3x_4-3x_1x_3+x_1-x_3x_4+x_3-x_4,\\
    f_3(x_1,x_2,x_4)&=x_1x_4-x_1-2x_4+1,\\
    f_4(x_1,x_2,x_3)&=-x_1x_2x_3+3x_1x_3-x_2x_3+x_2-1.
  \end{align*}

  We show using Lemma~\ref{lmm:ne-poly-cond} that if $x_i\in\s{0,1}$ for some $i$ then 
  $x\notin\NE(G_4)$. There are 8 cases.

  \paragraph{Case 1: $x_1=0$} then $f_2=-x_3x_4+x_3-x_4,~f_3=-2x_4+1,~f_4=-x_2x_3+x_2-1$.
  If $x_4=\frac{1}{2}$ then $f_2=\frac{x_3-1}{2}$. If $x_3=1$ then $f_4=-1<0$ so $x_4=0$ -- 
  a contradiction. Therefore $x_3<1$ so $f_2<0$ so $x_2=0$. Since $x_4=\frac{1}{2}$ we get
  $0=f_4=-1$ -- a contradiction.

  Therefore $x_4\ne\frac{1}{2}$. If $x_4<\frac{1}{2}$ then $f_3>0$ so $x_3=1$ so $f_4=-1$ 
  so $x_4=0$ so $f_2=1$ so $x_2=1$. Therefore $x=(0,1,1,0)\notin\PNE(G_4)$. If 
  $x_4>\frac{1}{2}$ then $f_3<0$ so $x_3=0$ so $f_2=-x_4<0$ so $x_2=0$ so $f_4=-1$ so 
  $x_4=0$ -- a contradiction.

  Overall $x_4$ cannot take any value in $[0,1]$, so $x\notin\NE(G_4)$.

  \paragraph{Case 2: $x_1=1$} then 
  $f_2=2x_3x_4-2x_3-x_4+1,~f_3=-x_4,~f_4=-2x_2x_3+x_2+3x_3-1$.
  If $x_4>0$ then $f_3<0$ so $x_3=0$ so $f_2=-x_4+1$ and $f_4=x_2-1$. If $x_4<1$ then 
  $f_2>0$ so $x_2=1$ so $f_1=-1$ so $x_1=0$ -- a contradiction. Therefore, still assuming 
  that $x_4>0$ we get $x_4=1$ so $f_4=x_2-1\ge0$ so $x_2=1$ so $f_1=-1$ so $x_1=0$ -- a 
  contradiction.

  Therefore $x_4=0$. So $f_2=-2x_3+1$. If $x_3=\frac{1}{2}$ then $f_4=\frac{1}{2}$ so 
  $x_4=1$ -- a contradiction. Therefore $x_3\ne\frac{1}{2}$. If $x_3>\frac{1}{2}$ then 
  $f_2<0$ so $x_2=0$ so $f_4=3x_3-1>0$ so $x_4=1$ -- a contradiction. Therefore 
  $x_3<\frac{1}{2}$. Then $f_2>0$ so $x_2=1$ so $f_1=2x_3-1<0$ so $x_1=0$ -- a 
  contradiction.

  Overall $x_1\ne1$.

  \paragraph{Case 3: $x_2=0$} then $f_1=-2x_3x_4+1,~f_4=3x_1x_3-1$. By cases 1 and 2, 
  $x_1\in(0,1)$ so $f_1=0$ so $x_3x_4=\frac{1}{2}$ so $x_3,x_4\ne0$ so $f_3,f_4\ge0$.  
  Therefore $x_3=\frac{1}{2x_4}$ so $f_4=\frac{3x_1}{2x_4}-1\ge0$ so 
  $x_1\ge\frac{2x_4}{3}$. Also, $f_3=x_1x_4-x_1-2x_4+1\ge0$ so $(1-x_4)x_1\le1-2x_4$. If 
  $x_4=1$ then $f_3=-1<0$ -- a~contradiction. So $x_4<1$, so
  $\frac{2}{3}x_4\le x_1\le\frac{1-2x_4}{1-x_4}$. The quadratic inequality is then 
  $2x_4^2-8x_4+3\ge0$, and in particular $x_4<\frac{1}{2}$ so $x_3>1$ -- a contradiction to 
  $x_3\le1$.

  Overall $x_4$ cannot take on any value, concluding this case.

  \paragraph{Case 4: $x_2=1$} then $f_1=-x_3x_4+2x_3-1,~f_4=2x_1x_3-x_3$. Again $f_1=0$ so 
  $x_3=\frac{1}{2-x_4}$. If $x_4=1$ then $x_3=1$ and $f_3=-1$ so $x_3=0$ -- a 
  contradiction. Therefore $x_4<1$ so $x_3<1$ so $f_3,f_4\le0$.
  Also $x_3>0$ and by $f_4\le0$ we get $x_1\le\frac{1}{2}$.
  By $f_3\le0$ we get $x_1\ge\frac{1-2x_4}{1-x_4}$, so in total $1-x_4\ge2-4x_4$ so 
  $x_4\ge\frac{1}{3}$. Therefore $x_3\ge\frac{3}{5}$. This means that $f_3=f_4=0$, so 
  $x_1=\frac{1}{2}$ so $f_3=-\frac{3}{2}x_4+\frac{1}{2}=0$ so $x_4=\frac{1}{3}$ so 
  $x_3=\frac{3}{5}$. Therefore $f_2=-\frac{1}{30}$ so $x_2=0$ -- a contradiction to 
  $x_2=1$. This establishes case 4. 

  \paragraph{Case 5: $x_3=0$} then $f_1=-2x_2+1,~f_2=x_1-x_4,~f_4=x_2-1$. By the previous 
  cases $x_1,x_2\in(0,1)$ so $f_1=f_2=0$ so $x_2=\frac{1}{2}$ and $x_1=x_4$. We also get 
  $f_4<0$ so $x_4=0$ so $x_1=0$ -- a contradiction.

  \paragraph{Case 6: $x_3=1$} then 
  $f_1=x_2x_4-2x_4+1,~f_2=3x_1x_4-2x_1-2x_4+1,~f_4=-x_1x_2+3x_1-1$.
  Again $f_1=f_2=0$ so $x_4=\frac{1}{2-x_2}$ so 
  $f_2=\frac{3x_1}{2-x_2}-2x_1-\frac{2}{2-x_2}+1=0$ so 
  $\frac{2x_2-1}{2-x_2}x_1=\frac{x_2}{2-x_2}$ so $(2x_1-1)x_2=1$ so $x_2>1$ because 
  $x_1\in(0,1)$ -- a contradiction to $x_2<1$.

  \paragraph{Case 7: $x_4=0$} then $f_3=-x_1+1$. By the previous cases, 
  $x_1,x_2,x_3\in(0,1)$ so $f_3=0$ so $x_1=1$ -- a contradiction.

  \paragraph{Case 8: $x_4=1$} then $f_2=x_1-1$. Like in case 7 we get $x_1\in(0,1)$ and 
  $x_1=1$ -- a contradiction.
  
  \bigskip

  Therefore the only NEs are fully mixed. This means that $x\in\NE(G_4)$ satisfies 
  $f_i(x)=0$ for every $i$. We solve this now.

  Let $F=\s{f_1,f_2,f_3,f_4}$. It is a set of polynomials in the polynomial ring 
  $\Q[x_1,x_2,x_3,x_4]$. We begin by finding the Gr\"obner basis $G=\s{g_1,g_2,g_3,g_4}$ of 
  the ideal $\langle F\rangle$, with the monomial ordering generated lexicographically by 
  $x_1\succ x_2\succ x_3\succ x_4$, given by
  \begin{align*}
    g_1&=5x_4^6-44x_4^5+143x_4^4-163x_4^3+85x_4^2-21x_4+2,\\
    g_2&=28x_3-1465x_4^5+12302x_4^4-36947x_4^3+32897x_4^2-11699x_4+1447,\\
    g_3&=4x_2+255x_4^5-2094x_4^4+6053x_4^3-4687x_4^2+1401x_4-149,\\
    g_4&=7x_1+5x_4^5-39x_4^4+104x_4^3-59x_4^2+26x_4-9.
  \end{align*}
  These were found using Mathematica~\citep{Mathematica}.

  Now we compute the solutions of $\forall i~g_i=0$ to find those of $\forall i~f_i=0$.  
  First, we approximate the 6 roots of $g_1$ numerically:
  \begin{align*}
    r_1&\approx0.3200651976,~r_2\approx0.4231894254,\\
    r_3&\approx0.4135410939-0.09991760306i,~r_4\approx0.4135410939+0.09991760306i,\\
    r_5&\approx3.614831595-1.802444362i,~r_6\approx3.614831595+1.802444362i.
  \end{align*}
  $x_4\in(0,1)$ so $x_4\in\s{r_1,r_2}$. If $x_4=r_2$ then for every $r'_2\in[0.42,0.43]$, $g_2(r'_2,x_3)=0$ iff $x_3>1$ -- a contradiction to $x_3<1$.
  Therefore $x_4=r_1$.\footnote{See \ref{app:sturm} for how to approximate the roots of $g_1$ using exact arithmetic without relying on the problematic floating point arithmetic.} Using the equalities $g_2=g_3=g_4=0$ we arrive at the approximations
  \begin{align*}
    x_1\approx0.529270752820,&~x_2\approx0.846414728986,\\
    x_3\approx0.523440476515,&~x_4\approx0.320065197645.
  \end{align*}
  
  We now show that $x_1,x_2,x_3,x_4$ are irradical.
  First, for each of the 3 monomial orderings generated lexicographically by $x_2\succ 
  x_3\succ x_4\succ x_1$, $x_3\succ x_4\succ x_1\succ x_2$ and $x_4\succ x_1\succ x_2\succ 
  x_3$, we find corresponding Gr\"obner bases, in addition to the 
  one that we found initially. From them we collect the univariate polynomials
  \begin{align*}
    P_1(y)&=7y^6-42y^5+89y^4-83y^3+40y^2-10y+1,\\
    P_2(y)&=4y^6-27y^5+70y^4-79y^3+45y^2-13y+1,\\
    P_3(y)&=140y^6-511y^5+701y^4-454y^3+141y^2-19y+1,\\
    P_4(y)&=5y^6-44y^5+143y^4-163y^3+85y^2-21y+2.
  \end{align*}
  which we further know satisfy $P_i(x_i)=0$. These exist by the elimination theorem for 
  Gr\"obner bases \citep{cox1997ideals}.

  Now, we use Murty's criterion~\citep{murty2002prime} to show that $P_1,P_2,P_3,P_4$ are 
  all irreducible over $\Q$.
  \begin{enumerate}
    \item $P_1$ is irreducible because $P_1(18)=167595301$ is prime and the $H$ value of 
      $P_1$ is $\frac{89}{7}\approx12.714$.
    \item $P_2$ is irreducible because $P_2(28)=1504207909$ is prime and the $H$ value of 
      $P_2$ is $\frac{79}{4}=19.75$.
    \item $P_3$ is irreducible because $P_3(11)=175397399$ is prime and the $H$ value of 
      $P_3$ is $\frac{701}{140}\approx5.007$.
    \item $P_4$ is irreducible because $P_4(39)=13945135583$ is prime and the $H$ value of 
      $P_4$ is $\frac{163}{5}=32.6$.
  \end{enumerate}
  
  As a result, $P_i$ is the minimal polynomial of $x_i$ (up to normalizing the leading 
  coefficient of $P_i$), so $x_i$ is irrational for every $i$. Now, it can be verified that 
  the Galois group of $P_i$ is $S_6$ for every $i\in\s{1,2,3,4}$ (we used 
  Magma~\citep{MR1484478}). Therefore $x_1,x_2,x_3,x_4$ are all irradical.
  This concludes the proof.
\end{proof}

\section{Five players}\label{sec:5}

Define the $2\times2\times2\times2\times2$ game $G_5$:
\begin{figure}[H]
\centering
\begin{align*}
  \left[\begin{array}{ccccc|ccccc}
    0&0&0&0&0&1&0&0&0&1\\
    0&0&0&0&1&1&1&0&0&0\\
    0&0&0&1&0&1&0&0&1&0\\
    0&0&0&1&1&0&0&1&0&0\\
    0&0&1&0&0&1&0&0&0&2\\
    0&0&1&0&1&0&1&1&0&0\\
    0&0&1&1&0&1&1&0&1&0\\
    0&0&1&1&1&0&1&0&0&1\\
    0&1&0&0&0&1&0&0&0&0\\
    0&1&0&0&1&0&0&1&1&1\\
    0&1&0&1&0&1&1&0&1&1\\
    0&1&0&1&1&0&1&1&0&1\\
    0&1&1&0&0&1&1&1&0&0\\
    0&1&1&0&1&1&0&1&0&1\\
    0&1&1&1&0&1&0&2&1&1\\
    0&1&1&1&1&0&0&1&0&1\\
    1&0&0&0&0&0&0&0&0&1\\
    1&0&0&0&1&0&1&0&1&0\\
    1&0&0&1&0&0&0&0&1&0\\
    1&0&0&1&1&1&0&1&0&0\\
    1&0&1&0&0&0&0&0&0&2\\
    1&0&1&0&1&1&1&1&0&0\\
    1&0&1&1&0&1&1&0&1&0\\
    1&0&1&1&1&0&1&0&0&1\\
    1&1&0&0&0&0&0&0&0&0\\
    1&1&0&0&1&1&0&1&1&1\\
    1&1&0&1&0&0&1&0&1&1\\
    1&1&0&1&1&1&1&1&0&1\\
    1&1&1&0&0&0&1&1&0&0\\
    1&1&1&0&1&1&0&0&1&1\\
    1&1&1&1&0&0&0&1&1&0\\
    1&1&1&1&1&1&0&1&0&1
  \end{array}\right]
\end{align*}
\caption{The game $G_5$}
\label{fig:G5}
\end{figure}

\begin{proposition}\label{prop:n=5}
  $G_5$ has a unique NE $((x_1,1-x_1),(x_2,1-x_2),(x_3,1-x_3),(x_4,1-x_4),(x_5,1-x_5))$ 
  given by
  \begin{align*}
    x_1\approx0.350370422,&~x_2\approx0.646516479,~x_3\approx0.648711818,\\
    x_4\approx0.371748770,&~x_5\approx0.368061687.
  \end{align*}
  The probabilities in the NE are all irrational and irradical.
\end{proposition}
\begin{proof}
  First observe that $\PNE(G_4)=\emptyset$ using this table of profitable deviations:
  {\footnotesize
  \begin{align*}
    \begin{array}{|c|c|}
      \hline
      \text{action profile}&\text{unsatisfied players}\\
      \hline
      (0,0,0,0,0)&4\\
      (0,0,0,0,1)&3,5\\
      (0,0,0,1,0)&2\\
      (0,0,0,1,1)&1,2\\
      (0,0,1,0,0)&2,4\\
      (0,0,1,0,1)&1,5\\
      (0,0,1,1,0)&5\\
      (0,0,1,1,1)&3\\
      (0,1,0,0,0)&3,4,5\\
      (0,1,0,0,1)&1,2\\
      (0,1,0,1,0)&3\\
      (0,1,0,1,1)&1,4\\
      (0,1,1,0,0)&4,5\\
      (0,1,1,0,1)&2\\
      (0,1,1,1,0)&2\\
      (0,1,1,1,1)&1,2\\
      (1,0,0,0,0)&1,4\\
      (1,0,0,0,1)&1,3,5\\
      (1,0,0,1,0)&1,2\\
      (1,0,0,1,1)&2,4\\
      (1,0,1,0,0)&1,2,4\\
      (1,0,1,0,1)&5\\
      (1,0,1,1,0)&5\\
      (1,0,1,1,1)&3\\
      (1,1,0,0,0)&1,3,4,5\\
      (1,1,0,0,1)&2\\
      (1,1,0,1,0)&1,3\\
      (1,1,0,1,1)&4\\
      (1,1,1,0,0)&1,4,5\\
      (1,1,1,0,1)&2,3\\
      (1,1,1,1,0)&1,2,5\\
      (1,1,1,1,1)&2,4\\
      \hline
    \end{array}
  \end{align*}
  }
  
  Now the polynomials $f_1,f_2,f_3,f_4,f_5$ from Equation~\ref{eq:polys} are
  \begin{align*}
    f_1(x_2,x_3,x_4,x_5)=&-5x_2x_3x_4x_5+4x_2x_3x_4+2x_2x_3x_5-x_2x_3+3x_2x_4x_5\\
    &-2x_2x_4-2x_2x_5+x_2+x_3x_4x_5-x_3x_4-x_4x_5+x_4\\
    &+2x_5-1,\\
    f_2(x_1,x_3,x_4,x_5)=&~x_3x_4x_5+2x_3x_4-2x_3-2x_4x_5+1,\\
    f_3(x_1,x_2,x_4,x_5)=&-2x_1x_2x_4x_5+x_1x_2x_4+x_1x_2x_5+2x_1x_4x_5-x_1x_4-x_1x_5\\
    &+3x_2x_4x_5-3x_2x_4+x_2-x_4x_5+x_4-x_5,\\
    f_4(x_1,x_2,x_3,x_5)=&~2x_1x_2x_3x_5-2x_1x_2x_3-x_1x_2x_5+x_1x_2-x_1x_3x_5+x_1x_3\\
    &+x_1x_5-x_1-x_2x_3x_5+x_2x_3+x_2x_5-x_2-2x_5+1,\\
    f_5(x_1,x_2,x_3,x_4)=&-x_1x_2x_3x_4+x_1x_2x_3+x_1x_2x_4-x_1x_2+x_1x_3x_4-x_1x_3\\
    &-x_1x_4+x_1-x_2x_3x_4+3x_2x_4-x_3x_4+x_3-1.
  \end{align*}

  We show using Lemma~\ref{lmm:ne-poly-cond} that if $x_i\in\s{0,1}$ for some $i$ then 
  $x\notin\NE(G_5)$. The proof proceeds as follows. For each $i$ in turn, we fix $x_i \in \{ 0, 1 \}$ and consider the subgame between the 4 players obtained by restricting player $i$ to play the pure action corresponding to $x_i$. We then perform a case analysis to find all profiles $x_{-i}$ for these 4 players which are NEs in this subgame. The case analysis is similar to the case analysis in the proof of Proposition~\ref{prop:n=4}. We then show that for each such NE profile $x_{-i}$, $(2x_i-1)f_i(x_{-i})<0$, meaning~$i$ is not playing a best response to $x_{-i}$ in the full $5$-player game $G_5$. So for player $i$ to play a best response in any NE of the full $5$-player game player, $i$ will need to choose $x_i\in(0,1)$.
  
  There are 10 cases in the proof (one for each $i\in[5]$ and $x_i\in\s{0,1}$). Here is additional information regarding these 10 cases (here the approximate values
  were computed using Mathematica):

  \paragraph{Case 1: $x_1=0$} Here the unique partial NE for players $2,3,4,5$ is given by
  \begin{align*}
    x_2&\approx0.650518016106,~x_3\approx0.638238319763,\\
    x_4&\approx0.402794248582,~x_5\approx0.433321011106.
  \end{align*}
  And $f_1(x_{-1})\approx0.103778882911>0$. Therefore $x_1\ne0$ (observe that $\abs{\nabla 
  f_1}\le1000$ so the value of $f_1$ at the exact value of $x_{-1}$ will still be $>0$ when 
  using 12-digit approximations).

  \paragraph{Case 2: $x_1=1$} Here the unique partial NE for players $2,3,4,5$ is given by
  \begin{align*}
    x_2&\approx0.589697169563,~x_3\approx0.681923203912,\\
    x_4&\approx0.338247172171,~x_5\approx0.218624870529.
  \end{align*}
  And $f_1(x_{-1})\approx-0.245843038086<0$. Therefore $x_1\in(0,1)$.

  \paragraph{Case 3: $x_2=0$} There is no partial NE where $x_1\in(0,1)$, so $x_2\ne0$.

  \paragraph{Case 4: $x_2=1$} Here all partial NEs for players $1,3,4,5$ are given by
  \begin{align*}
    x_1\in(0,1),~x_3=1,~x_4=0,~x_5=\frac{1}{2}.
  \end{align*}
  And $f_2(x_{-2})=-1<0$. Therefore $x_2\in(0,1)$.

  \paragraph{Case 5: $x_3=0$} There is no partial NE where $x_1,x_2\in(0,1)$, so 
  $x_3\ne0$.

  \paragraph{Case 6: $x_3=1$} There is no partial NE where $x_1,x_2\in(0,1)$, so 
  $x_3\in(0,1)$.

  \paragraph{Case 7: $x_4=0$} There is no partial NE where $x_1,x_2,x_3\in(0,1)$, so 
  $x_4\ne0$.

  \paragraph{Case 8: $x_4=1$} There is no partial NE where $x_1,x_2,x_3\in(0,1)$, so 
  $x_4\in(0,1)$.

  \paragraph{Case 9: $x_5=0$} There is no partial NE where $x_1,x_2,x_3,x_4\in(0,1)$, so 
  $x_5\ne0$.

  \paragraph{Case 10: $x_5=1$} There is no partial NE where $x_1,x_2,x_3,x_4\in(0,1)$, so 
  $x_5\in(0,1)$.

  \bigskip

  Therefore the only NEs are fully mixed. This means that $x\in\NE(G_5)$ satisfies 
  $f_i(x)=0$ for every $i$. We solve this now.

  Let $F=\s{f_1,f_2,f_3,f_4,f_5}$. It is a set of polynomials in the polynomial ring 
  $\Q[x_1,x_2,x_3,x_4,x_5]$. We begin by finding the Gr\"obner basis 
  $G=\s{g_1,g_2,g_3,g_4,g_5}$ of the ideal $\langle F\rangle$, with the monomial ordering 
  generated lexicographically by $x_1\succ x_2\succ x_3\succ x_4\succ x_5$, given in full 
  in \ref{app:grobner-basis} (using Mathematica).

  We now compute the solutions of $\forall i~g_i=0$ to find those of $\forall i~f_i=0$. The proof proceeds as follows. We observe that like in the proof of Proposition~\ref{prop:n=4}, $g_1\in\Q[x_5]$ and for $i\ge2$, $g_i=a_ix_{6-i}+h_i(x_5)$ where $a_i\in\Z$ and $h_i\in\Q[x_5]$. So we first find the roots of $g_1$, using a similar argument to the one in the proof of Proposition~\ref{prop:n=4} and \ref{app:sturm}, and find $5$ roots of $g_5$ in $(0,1)$, meaning $5$ possible values for $x_5$. Then for each of the latter four we show using $g_2,\ldots,g_5$ that one of $x_1,x_2,x_3,x_4$ falls outside $(0,1)$, implying a contradiction. In conclusion, there is a unique possible value for $x_5$ and by the structure of $g_2,\ldots,g_5$ a unique possible value for each of $x_1,x_2,x_3,x_4$.
  This solution is given in the 
  approximated decimal form:
  \begin{align*}
    x_1\approx0.3503704221,&~x_2\approx0.6465164785,~x_3\approx0.6487118183,\\
    x_4\approx0.3717487703,&~x_5\approx0.3680616872.
  \end{align*}

  To see that $x_1,x_2,x_3,x_4,x_5$ are irradical, as before, we find their minimal 
  polynomials (up to scaling):
  \begin{align*}
    P_1(y)=&~576 y^{26}+7360 y^{25}+4496 y^{24}+250816 y^{23}+355564 y^{22}\\
    &-7898280 y^{21}+7658244 y^{20}+75548910 y^{19}-142993432 y^{18}\\
    &-189077220 y^{17}+334832314 y^{16}+474658874 y^{15}+1696776519 y^{14}\\
    &-7562033350 y^{13}+1512397109 y^{12}+14656660866 y^{11}\\
    &-10470778075 y^{10}-3782764895 y^9-665928467 y^8+10050992181 y^7\\
    &-5447529632 y^6-1586217407 y^5+562969676 y^4+1493433257 y^3\\
    &-79516172 y^2-522298360 y+132583232,
  \end{align*}
  \begin{align*}
    P_2(y)=&~2057728 y^{26}-47527232 y^{25}+374368064 y^{24}-482543600 y^{23}\\
    &-11253727536 y^{22}+90191808584 y^{21}-352161035060 y^{20}\\
    &+829456047656 y^{19}-1120735934348 y^{18}+286387243298 y^{17}\\
    &+2393470340090 y^{16}-6252418753985 y^{15}+9238553888534 y^{14}\\
    &-9587641717941 y^{13}+7391393209142 y^{12}-4277517478697 y^{11}\\
    &+1815337752171 y^{10}-515212602860 y^9+59595129244 y^8\\
    &+25514336227 y^7-17322169528 y^6+5510430025 y^5-1139164516 y^4\\
    &+160112123 y^3-14874613 y^2+828146 y-20988,
  \end{align*}
  \begin{align*}
    P_3(y)=&~5828 y^{26}-80590 y^{25}+471147 y^{24}-1473516 y^{23}\\
    &+1995893 y^{22}+3280961 y^{21}-21791522 y^{20}+51425278 y^{19}\\
    &-93080861 y^{18}+203283288 y^{17}-444991348 y^{16}+713613468 y^{15}\\
    &-837466118 y^{14}+925602099 y^{13}-1210417319 y^{12}+1552957912 y^{11}\\
    &-1585613560 y^{10}+1241271492 y^9-772369636 y^8+401799920 y^7\\
    &-180281904 y^6+69151344 y^5-21721184 y^4+5260240 y^3\\
    &-909264 y^2+99200 y-5120,
  \end{align*}
  \begin{align*}
    P_4(y)=&~367535904 y^{26}-4069658144 y^{25}+15863607440 y^{24}+4840925000 y^{23}\\
    &-332049318712 y^{22}+1798491222250 y^{21}-5880537317282 y^{20}\\
    &+13925726681306 y^{19}-25570689763592 y^{18}+37716767814324 y^{17}\\
    &-45641117387319 y^{16}+45922091951685 y^{15}-38746001553780 y^{14}\\
    &+27551639828330 y^{13}-16547091697878 y^{12}+8389403526528 y^{11}\\
    &-3578413101941 y^{10}+1274916123621 y^9-374578292008 y^8\\
    &+88741617902 y^7-16260708632 y^6+2104604912 y^5-141853168 y^4\\
    &-7308320 y^3+2769472 y^2-256128 y+8192,
  \end{align*}
  \begin{align*}
    P_5(y)=&~25772032 y^{26}-399987456 y^{25}+2634374272 y^{24}-8416506944 y^{23}\\
    &+2215910496 y^{22}+110722637568 y^{21}-612619393968 y^{20}\\
    &+2025933659884 y^{19}-4916502391844 y^{18}+9383743371458 y^{17}\\
    &-14566831934444 y^{16}+18743281388994 y^{15}-20217905397986 y^{14}\\
    &+18405783324440 y^{13}-14192403999687 y^{12}+9280251573089 y^{11}\\
    &-5141465321619 y^{10}+2406410117193 y^9-946553896669 y^8\\
    &+310404824295 y^7-83872244335 y^6+18358684229 y^5\\
    &-3175419088 y^4+417903630 y^3-39340244 y^2+2360724 y-67892.
  \end{align*}

  These can all be verified to be irreducible using Mathematica. Then, 
  the Galois groups of each of these 5 polynomials can be verified to be all $S_{26}$ which 
  is unsolvable (using e.g., Magma). Therefore $x_1,x_2,x_3,x_4,x_5$ are 
  all irradical. This concludes the proof.
\end{proof}

\section{Six players and up}\label{sec:ext}

Here we use the constructions of Section~\ref{sec:4} and Section~\ref{sec:5} to construct 
games $G_n$ for every $n\ge6$ with the same property as $G_4$ and $G_5$. This completes the 
picture in terms of the (general worst case) inexpressibility of NEs with radicals.

First, we use the notion of \emph{product of games}:

\begin{definition}
  Given are games $H^1,H^2$ such that $H^j$ has $n_j$ players, the pure action 
  profiles $A^j$ and the payoff functions $u^j_1,\ldots,u^j_{n_j}$. Define the 
  \emph{product game} $H^1\times H^2$ as follows:
  \begin{itemize}
    \item there are $n_1+n_2$ players,
    \item the set of pure action profiles is $A^1\times A^2$,
    \item the payoff functions are $u_i(a^1,a^2)=u^1_i(a^1)$ if $i\in[n_1]$ and 
      $u^2_i(a^2)$ if $i\in\s{n_1+1,\ldots,n_1+n_2}$.
  \end{itemize}
  If $H_1=H_2$, also write $H_1\times H_2=H_1^{\times2}$, and generalize the notation to a 
  product of $q$ copies of $H_1$ by $H_1^{\times q}$.
\end{definition}

This operation is associative, therefore $H_1^{\times q}$ is well defined. Here is a property of $H^1\times H^2$ that follows by definition:

\begin{lemma}\label{lmm:product-ne}
  $\NE(H^1\times H^2)=\NE(H^1)\times\NE(H^2)$.
\end{lemma}

Using this lemma, for every $n=4n_1+5n_2$ where $n_1,n_2\in\N$ are not both $0$, the game $G_4^{\times n_1}\times G_5^{\times n_2}$ has a unique NE and its probabilities are all irradical (see Proposition~\ref{prop:n>=6} below). However, these constructions do not cover the cases 
$n=6,7,11$. To fill these holes, we now show how to add 2 players to each of  $G_4$ and 
$G_5$ while maintaining their desired properties (a unique NE, whose probabilities are 
irradical).

Define the $2\times2\times2$ game $H_3$:
\begin{figure}[H]
\centering
\begin{align*}
  \left[\begin{array}{ccc|ccc}
    0&0&0&0&1&1\\
    0&0&1&0&1&1\\
    0&1&0&0&1&0\\
    0&1&1&0&0&1\\
    1&0&0&0&0&1\\
    1&0&1&0&0&0\\
    1&1&0&0&1&0\\
    1&1&1&0&0&0
  \end{array}\right]
\end{align*}
\caption{The game $H_3$}
\label{fig:H3}
\end{figure}

Here is the main property of interest of this game:

\begin{proposition}\label{prop:H3-NE}
  Let $x=((x_1,1-x_1),(x_2,1-x_2),(x_3,1-x_3))$ such that $x_1\in(0,1)$. Then $x\in\NE(H_3)$ iff $x_1=x_2=x_3$.
\end{proposition}
\begin{proof}
  First, the polynomials $f_1,f_2,f_3$ from Equation~\ref{eq:polys} of $H_3$ are
  \begin{align*}
    f_1(x_2,x_3)=0,~f_2(x_1,x_3)=x_1-x_3,~f_3(x_1,x_2)=x_2-x_1.
  \end{align*}

  By definition, for every $x_1\in(0,1)$ setting $x_1=x_2=x_3$ means that $f_1(x)=f_2(x)=f_3(x)=0$, meaning that $x\in\NE(H_3)$. Now we show that if $x_1\in(0,1)$ and $x\in\NE(H_3)$ then $x_1=x_2=x_3$.

  We use Lemma~\ref{lmm:ne-poly-cond}. $f_1=0$ so $x_1$ can take on 
  any value, regardless of the values of $x_2,x_3$. Assume that $x_1\in(0,1)$.
  Suppose 
  for contradiction that $x_2=0$. Then $f_3=-x_1<0$, so $x_3=0$ so $f_2=x_1>0$ so $x_2=1$ 
  -- a contradiction. Similarly, if $x_2=1$, then $f_3=1-x_1>0$ so $x_3=1$ so $f_2=x_1-1<0$ 
  so $x_2=0$ -- again a contradiction. Therefore $x_2\in(0,1)$ so $f_2=0$ so $x_3=x_1$ so $f_3=0$ so $x_1=x_2=x_3$.
\end{proof}

Next we define a variation of the product game with a small overlap in the players.  

\begin{definition}
  Let $G$ be a game such that
  \begin{itemize}
    \item there are $n$ players,
    \item each player has $2$ pure actions,
    \item the payoffs are given by functions $u_1,\ldots,u_n$,
    \item there is a unique NE, and this NE is fully mixed.
  \end{itemize}
  Define the game $G\circ H_3$ as follows:
  \begin{itemize}
    \item there are $n+2$ players,
    \item each player has $2$ pure actions,
    \item the payoff functions are given by $u_i(a,a_{n+1},a_{n+2})=u_i(a)$ for $i\in[n]$, 
      and for $i\in\s{n+1,n+2}$ the payoff is defined as 
      $u_i(a,a_{n+1},a_{n+2})=u^{H_3}_{i-(n-1)}(a_n,a_{n+1},a_{n+2})$ (as defined in 
      Figure~\ref{fig:H3}).
  \end{itemize}
\end{definition}

The most important thing about this operation $\circ H_3$ is that it is comptaible with 
$G_4$ and $G_5$. The key property of $G\circ H_3$ is as follows:

\begin{lemma}\label{lmm:circ-ne}
  If $\NE(G)=\s{((x_1,1-x_1),\ldots,(x_n,1-x_n))}$ and $x_n\in(0,1)$ then
  \begin{align*}
    \NE(G\circ H_3)=\s{((x_1,1-x_1),\ldots,(x_n,1-x_n),(x_n,1-x_n),(x_n,1-x_n))}.
  \end{align*}
\end{lemma}
\begin{proof}
  By definition. Observe that a NE of $G\circ H_3$, restricted to players $n,n+1,n+2$, is a 
  NE to $H_3$, because the payoff of the first player in $H_3$ (meaning player $n$ in $G\circ H_3$) is constant. Since player $n$ plays $(x_n,1-x_n)$ in the only NE of $G$, player $n$ will need to play it in any NE of $G\circ H_3$. Therefore we 
  can invoke Proposition~\ref{prop:H3-NE} here to get the desired result.
\end{proof}

Now the games $G_n$ for $n\ge6$ follow. Let $n\ge6$. Write $n=4q+r$ where $1\le q\in\N$ and 
$r\in\s{0,1,2,3}$. Define
\begin{align}
  G_n=\begin{cases}
    G_4^{\times q}&r=0,\\
    G_4^{\times q-1}\times G_5&r=1,\\
    G_4^{\times q-1}\times(G_4\circ H_3)&r=2,\\
    G_4^{\times q-1}\times(G_5\circ H_3)&r=3.
  \end{cases}
\end{align}

\begin{proposition}\label{prop:n>=6}
  For every $n\ge6$, the game $G_n$ has a unique NE, and the probabilities in this NE are 
  all irradical.
\end{proposition}
\begin{proof}
  Let $n=4q+r$. By Proposition~\ref{prop:n=4}, Proposition~\ref{prop:n=5} and 
  Lemma~\ref{lmm:circ-ne} we obtain that $G_4\circ H_3$ and $G_5\circ H_3$ both have a 
  unique NE, and that NE is irradical in both games. Now, for any $n\ge6$, $G_n$ is a 
  product of games which all have a unique NE, and that NE is irradical in all those games. Therefore each $G_n$ has a unique NE, and its probabilities are all irradical.
\end{proof}

\section*{Acknowledgements}
We thank Teo Collins for the reference to Gr\"obner bases.

\appendix

\section{Further details in the proof of Proposition~\ref{prop:n=4}}\label{app:sturm}

In the proof of Proposition~\ref{prop:n=4} we use the approximations of the roots of $g_1$ to determine the unique value of $x_4$. To circumvent possible floating point issues, we can use the Sturm sequence of $g_1$ to determine using exact arithmetic that (1) there are exactly 2 real roots, and (2) one is in the interval $(0.3,0.4]$ and one in the interval $(0.4,0.5]$ (\cite{sturm1835memoire}). The Sturm sequence of $g_1$ can be found using Mathematica:
\begin{align*}
    p_0=&g_1=5x_4^6-44x_4^5+143x_4^4-163x_4^3+85x_4^2-21x_4+2,\\
    p_1=&g_1'=30x_4^5-220x_4^4+572x_4^3-489x_4^2+170x_4-21,\\
    p_2=&-\remn(p_0,p_1)=\frac{55}{9}x_4^4-\frac{5249}{90}x_4^3+\frac{943}{15}x_4^2-\frac{433}{18}x_4+\frac{47}{15},\\
    p_3=&-\remn(p_1,p_2)=-\frac{27110403}{30250}x_4^3+\frac{15927363}{15125}x_4^2-\frac{2514591}{6050}x_4\\
    &+\frac{831852}{15125},\\
    p_4=&-\remn(p_2,p_3)=\frac{4813052861375}{81663772313601}x_4^2-\frac{1017283844500}{27221257437867}x_4\\
    &+\frac{417696219875}{81663772313601},\\
    p_5=&-\remn(p_3,p_4)=\frac{46666740312733677523326}{1531601841083641320125}x_4\\
    &-\frac{19799411241381912287163}{1531601841083641320125},\\
    p_6=&-\remn(p_4,p_5)=\frac{3506083202136869448018665125}{26667695965024814331677001308676}.
\end{align*}
where $\remn(f,g)$ is the remainder of the polynomial division of $f$ by $g$.

For every $a\in\R\cup\s{-\infty,\infty}$, let $V(a)$ be the number of sign changes in the sequence $(p_0,p_1,p_2,p_3,p_4,p_5,p_6)$ (and the sign of something like $p_0(\infty)$ is defined to be $+$). Sturm's theorem states that for every $-\infty\le a<b\le\infty$, the number of real roots in $(a,b]$ is equal to $V(a)-V(b)$.

So we have $V(-\infty)-V(\infty)=4-2=2$, so there are $2$ real roots. Note that
\begin{align*}
    p_0\p{\frac{3}{10}}&=\frac{161}{40000},\\
    p_1\p{\frac{3}{10}}&=-\frac{2751}{10000},\\
    p_2\p{\frac{3}{10}}&=\frac{371}{7500},\\
    p_3\p{\frac{3}{10}}&=\frac{26761959}{30250000},\\
    p_4\p{\frac{3}{10}}&=-\frac{258737930605}{326655089254404},\\
    p_5\p{\frac{3}{10}}&=-\frac{28996945737809045150826}{7658009205418206600625},\\
    p_6\p{\frac{3}{10}}&=\frac{3506083202136869448018665125}{26667695965024814331677001308676},
\end{align*}
and
\begin{align*}
    p_0\p{\frac{4}{10}}&=-\frac{4}{3125},\\
    p_1\p{\frac{4}{10}}&=\frac{27}{625},\\
    p_2\p{\frac{4}{10}}&=-\frac{4}{625},\\
    p_3\p{\frac{4}{10}}&=-\frac{236727}{1890625},\\
    p_4\p{\frac{4}{10}}&=-\frac{32955935705}{81663772313601},\\
    p_5\p{\frac{4}{10}}&=-\frac{5663575581442206389163}{7658009205418206600625},\\
    p_6\p{\frac{4}{10}}&=\frac{3506083202136869448018665125}{26667695965024814331677001308676},
\end{align*}
and
\begin{align*}
    p_0\p{\frac{5}{10}}&=\frac{1}{64},\\
    p_1\p{\frac{5}{10}}&=\frac{7}{16},\\
    p_2\p{\frac{5}{10}}&=-\frac{31}{360},\\
    p_3\p{\frac{5}{10}}&=-\frac{383139}{242000},\\
    p_4\p{\frac{5}{10}}&=\frac{380134673875}{326655089254404},\\
    p_5\p{\frac{5}{10}}&=\frac{28271671319879411796}{12252814728669130561},\\
    p_6\p{\frac{5}{10}}&=\frac{3506083202136869448018665125}{26667695965024814331677001308676}.
\end{align*}
So $V(0.3)-V(0.4)=4-3=1$ and $V(0.4)-V(0.5)=3-2=1$, so $(0.3,0.4]$ and $(0.4,0.5]$ each contain exactly one of the real roots. More specifically, $r_1\in(0.3,0.4]$ and $r_2\in(0.4,0.5]$.

From here the argument in the proof of Proposition~\ref{prop:n=4} that $x_4$ cannot be $r_2$ holds and the rest of the proof goes through.

As for $r_1$, the method above can be further used to verify that $r_1\in(0.320065197644,0.320065197645]$ using exact arithmetic, justifying the approximation for $x_4=r_1$ given in the proposition statement.

\section{The Gr\"obner basis in Proposition~\ref{prop:n=5}}\label{app:grobner-basis}

Here is the Gr\"obner basis $G$ used in Proposition~\ref{prop:n=5}. It is 
$\s{g_1,g_2,g_3,g_4,g_5}$ given as follows:
\begin{align*}
  g_1=&~25772032x_5^{26}-399987456x_5^{25}+2634374272x_5^{24}-8416506944x_5^{23}\\
  &+2215910496x_5^{22}+110722637568x_5^{21}-612619393968x_5^{20}\\
  &+2025933659884x_5^{19}-4916502391844x_5^{18}+9383743371458x_5^{17}\\
  &-14566831934444x_5^{16}+18743281388994x_5^{15}-20217905397986x_5^{14}\\
  &+18405783324440x_5^{13}-14192403999687x_5^{12}+9280251573089x_5^{11}\\
  &-5141465321619x_5^{10}+2406410117193x_5^9-946553896669x_5^8\\
  &+310404824295x_5^7-83872244335x_5^6+18358684229x_5^5-3175419088x_5^4\\
  &+417903630x_5^3-39340244x_5^2+2360724x_5-67892,
\end{align*}

{\tiny
\begin{align*}
  g_2=&-5229757014096166506106706852927054776656409436526688311934217666969182934601238598093813122993784633757481615360x_5^{25}\\
  &+77417348381807438139298712751795524837772256296658899465798071561263108113263168115187337389190007008413770618112x_5^{24}\\
  &-479003561865428975555899714436006958766026300658008374819291240222777323612217905003850065127110597318088478641792x_5^{23}\\
  &+1363516171631648731401937260410517759345140344024728313437486577525612859813544026537686042034050828432977603879232x_5^{22}\\
  &+533587922587820729708044335466442589479278858927142642702831828582442863250633819428985449118623257898991410037152x_5^{21}\\
  &-22100147438809308088861819847452779162469225691207064470932200568612066220139113476396459639356379722280870005175520x_5^{20}\\
  &+108454900964535452643280673745481810479068278931480402551733219020654394658331661285136003467695044177290332008077552x_5^{19}\\
  &-333073782919957459027278804921869980726997241882391079940534925956228969796725710815019469532991976447383544674526084x_5^{18}\\
  &+757624233673338301359440552712904560489502717271112905470949925790532905236334410472737236149940199890686132929621704x_5^{17}\\
  &-1357384404662525914311341693592403279296839308873698777339142031034705771862285501894438660803278932926999068794546582x_5^{16}\\
  &+1974990208789200599466586861637319326312916081007911794448182406682496223056735984172167459256026946263774303269130502x_5^{15}\\
  &-2374297193220515593964694404791585125075717016662628050023610781028665713690408237409030760913488268277385743966777032x_5^{14}\\
  &+2382354902550652961280745151068702383273085454800813608807850901474781766931357892102898412048158076148480110682671270x_5^{13}\\
  &-2006547338020399904478992449241580128008830980255928702874144615524529032076799032362502252260795483084720768354418906x_5^{12}\\
  &+1422303712594438667631688360360711116173856475208264246894179873753516371223301768725492398493859422183929726651437959x_5^{11}\\
  &-848591809572807446526737307028116087664196830302677461687679329565088044403583949846081471077208029772845733783978556x_5^{10}\\
  &+425256712357840429041584855108259194843259926786949852822199851920502756754006632376292656084188179152265324768387541x_5^9\\
  &-178190185696798219573888134723758197101657622267716428247687462654805037079392982664482126949337362000059358903070826x_5^8\\
  &+61968796201222942179577897911545319747966347541812098462585270467890006274705019861192475346963305894981804268851065x_5^7\\
  &-17686914368321420604395590031619931467353119221112708310058990785917431994157571306893482096694260412715975874006756x_5^6\\
  &+4075259952587509321223128820658917157429499391726104350291091494649891620567629499432147402492820100061680066097185x_5^5\\
  &-739721454661762561829150779036766386028574319580498394474729242346540997892430899746745978354273573458080278629282x_5^4\\
  &+101889556627033780067706963405180770731046840551112096527796778096348672520153243935935166739545857450254458665078x_5^3\\
  &-10014961581059325297534083799423222232193173182415278003356345107715054775890753680315781533615428759531894840428x_5^2\\
  &+626198151376926556516994281746232986281114459167528107589000512836306110760157779105130081244956865692258077872x_5\\
  &+14562790811532232105030294780908703577127741453867317051625330584746436330267698312201792323393195475584x_4\\
  &-18730395737658453954483839445489626955207312148851486189022226778682417235942672754625522902251804284010201836,
\end{align*}
}

{\tiny
\begin{align*}
  g_3=&~2334544751155071882026722780061349241642599026068676574946240154619529918561800486748917602409834307249946624x_5^{25}\\
  &-34826386272732604074820628503893799177194447217375096775483672338860118986103889174063050911927141610810379520x_5^{24}\\
  &+217633964950393711549186629464978437165300442561564386772405822217081697486573874407123129145535922950706091648x_5^{23}\\
  &-631004665463527900669117543672267152557837342795954164588602147543470354961220923336218777672617507047396540736x_5^{22}\\
  &-181187120952388236998386793367411756303633247262581447028534281445658265062871643279630089023039796374318343072x_5^{21}\\
  &+9925416763570921108480922821290611359261790440483739235088028728524852578911877914757663730293796516424536982112x_5^{20}\\
  &-49511025362638119989000055444640688180713101494312128922185056130200524969006756798221621831744228717747164071536x_5^{19}\\
  &+153606450244873998637400182694657307098185015597200929534563721055218708062776648371930100226217096327071566507812x_5^{18}\\
  &-352426714217615075454859339166993063141216123999988812644511442641474817930073511713523770604296026296813438396088x_5^{17}\\
  &+636541832234522264020461181180548709485760714008828971077784400214808511794157519019430471954449244547200377722358x_5^{16}\\
  &-933498979637348381777015122675301004004870810348162357840262569632788867775411763510678948565334319576360783673102x_5^{15}\\
  &+1131075337234367536418256190781035607571454240605206256230753284711414951414891505602860229050574408650155226782728x_5^{14}\\
  &-1143888525546125381293381264403258806739680192797154817504244626489887282710627211132657772617738453443292075022774x_5^{13}\\
  &+971114741160926999426755391516911794308562783195876583953124144282681326795227451967634738501132063379159602305106x_5^{12}\\
  &-693868399603588601921565117253658200432981408294055991494009684604764816981147591464337021509223488254637528974415x_5^{11}\\
  &+417307080155928782025068977202423092573461396691977026044607405101837819435313315500107795471975427014871318266976x_5^{10}\\
  &-210801786743647026314962428699955929372925777966515764581497272997599561283615798792354612095177686041986524450961x_5^9\\
  &+89033363860194520670813956986147529620606515277585298169282801928853396835629522791071359703663275943842956580202x_5^8\\
  &-31207237132427508982273789207629895352893344450292521048842395054212958515460240818669350564989492535868999725081x_5^7\\
  &+8976459942331587279270618598121344735704429300848888815877070229254977551136189717073239186833151007865116037992x_5^6\\
  &-2084149071494599066636840904363798736708176439312551310459408028002687362733692900701646271408658267467743099101x_5^5\\
  &+381158308756241245183778096001699736234731027675341960830001210231765721434485154435182600335913406210638306866x_5^4\\
  &-52889407110880592240935035564652467638860334385675951440377263860227391443955038430565826220210872540006884446x_5^3\\
  &+5236307109205295690183209426599989107958678066276665974271201847161057962269146837250130645744715913371810412x_5^2\\
  &-329731014995559194481139244161781866865235013681345347340507275899274691789786673963774969773751450999559968x_5\\
  &+3694744112928236745354448222098526709270851631110901155469041811943454903727724876142148905734577408x_3\\
  &+9931369572758577248575536574765581760354618753802907709675771557134507504523669268186354299397125834523196,
\end{align*}
}

{\tiny
\begin{align*}
  g_4=&~-3558500754024337082770180156322829498597886283821905562909626877933627437862004077143293724200219009927888896x_5^{25}\\
  &+53459533959080117233238158244911856852132196679058294422752524765582301160471501982607167430974092115094435584x_5^{24}\\
  &-337159203512668656410477904168737985969382049375082645168114655955720623507549883202440822289987871942813580672x_5^{23}\\
  &+994379773429930614855305353972967341994717800292461875312328526647809247215078010238428657356468774616503704256x_5^{22}\\
  &+189245805202958885164717457856217010089940221854661319802127336154003800666496337272960302229722564154334014816x_5^{21}\\
  &-15197125154505736973813689155185565411338769583659301109177897329934627331339207419412635177208389035828827853280x_5^{20}\\
  &+77035200685565232608193856879005729648408468927829648172264573836344951657512635824410473666817586520056592466384x_5^{19}\\
  &-241401429737158706348114011840171158855788925159564563866972792216940735637246580890337940955337976062511379408780x_5^{18}\\
  &+558630449888944438077465724771239908940298471522896091656782046424166365437557890923179651315268131895457834737840x_5^{17}\\
  &-1017224595597228136804728426231636929621054017646148023601819708031196666257616755854341239641312660582786669412250x_5^{16}\\
  &+1503817045296227892476950853512837221625294649034200024746679125649511207123301404686450154802628965251836163156638x_5^{15}\\
  &-1836897428293610264177371346711383956146293025290376833853869340646192152519382428015203859429983377866327285534440x_5^{14}\\
  &+1872995781108315067497312460676257927626533427837197474160535503357940379854809628626707176037342237011018857587474x_5^{13}\\
  &-1603373170906444757398030676224373499811872254724020941343536500445962959353227034242254754665909779095643373418578x_5^{12}\\
  &+1155288040878502319572608128049976975067745119574642103393020077529068588347437379307725141482112392993781461014925x_5^{11}\\
  &-700697376259091308302109508302940157504814446394086929947963510799016007546849283458571446545182585269186009462998x_5^{10}\\
  &+356934874693265870665365221583700170939530876165407760810376569278622878215127192877970425179784984014582661641105x_5^9\\
  &-152001087875514448869480222043342566350774439463857717641593193787967622803143243410107426653140398261285409508742x_5^8\\
  &+53706296084516238933941939388111267070100298333608164123229065510826751235963393953259745031831145831662722511511x_5^7\\
  &-15566872495760296770636202098010682330153678839979729723749724038484459853264945820724345917034624145600833219006x_5^6\\
  &+3640403667504735955909101560636890647144799421188362661421370193637130168818945991278667273963692398047273666453x_5^5\\
  &-670173581881674007325938393635716663699469008474575049226005766285927286886676689504925496451118751189529451670x_5^4\\
  &+93535104453294586615470233339118267735846441345555250692303574178503068296553438102544432800191989801053615838x_5^3\\
  &-9305270756368072254559737977820811715941760923436884753321555175707328683305744660190066221202432676682105108x_5^2\\
  &+588069610334459746429373597750058332341244517659467031174397469825967982566556498337804906924843284145942504x_5\\
  &+5095805281351607233898259232880569267178981710856112472144845244406570095429556031988777094079363968x_2\\
  &-17749305818938603275929612340121722447988753251997841330848673207984991309809988341490716053109866637501340,
\end{align*}
}

{\tiny
\begin{align*}
  g_5=&~9270243108944911855913423923940118860724549686336022978726911310961722458203883119715431043049984039600128x_5^{25}\\
  &-138175875266337145223434566252355090972050168617333616931556655001369528702606497111283046584897048123496192x_5^{24}\\
  &+862492337870286319208107141713916178120914949268939539624311844622498318270220300752453215856315985897987456x_5^{23}\\
  &-2495102461650519533506960742919876844791681341902817276675711486882107127757546342958357009723667467045595840x_5^{22}\\
  &-749636554146947543936421259610110568924854870730577569236289936020650163568330306595235407217786364530222432x_5^{21}\\
  &+39402580158298384618280673190243773364515744583137648398647346524048246636434435392106680856462315406779353056x_5^{20}\\
  &-196122144147700348424157453329717858016119647923907735554534428793476055982854327245096135183281374007838006224x_5^{19}\\
  &+607567407220738066950645872823562629356481100952338581837849090942575790518184039147606780264543360404788805388x_5^{18}\\
  &-1392034559331843275751880128686367355159492660859513756685269123763970049996344394913220744816481337236164688752x_5^{17}\\
  &+2510558512785233508346113941902852818446593588622592969503358098040942372525809496185640497523813014968232095706x_5^{16}\\
  &-3675774626200611204927149730882154016887117306216041250003754929446407301397549595440548848847486181249067607390x_5^{15}\\
  &+4445471379414986634335124565974011524961495924338229912795322982376585158552336715200597575424938671909289449640x_5^{14}\\
  &-4486118641224030633093582292243068588732437612076418746519037181965028002339068594596667317742675531315640964002x_5^{13}\\
  &+3798874308555711598482873640101602165951250009581540519626436478636365965265691902241294062878777541961423998546x_5^{12}\\
  &-2706192488036664027366752238840352044709794830418822829301155231922766756880774252399289499491242862301112642493x_5^{11}\\
  &+1621786402882167377760236931048484324618716111562222154462408173075352980667480821275949711332943322996801085942x_5^{10}\\
  &-815797787394998843496442733011612610801250048164205077534210923183660530362409166766004966247838070909131717313x_5^9\\
  &+342842754795609541453395389548504455515865006985724325769495898905537672061053294325586795901382310703074484590x_5^8\\
  &-119464200488048451770969329661905353520359137813236139439207882676912887567274947822290425167328332819674138255x_5^7\\
  &+34124598338008014536089722419271075905275461000608984295483205503028442881007977804213586157616567502762395862x_5^6\\
  &-7858371647506645013235958461172093630568069977126879574130925982176847618427475770739613780276848717890340861x_5^5\\
  &+1423375266711743719442643469712816450858096992327301545160933803057177533347744363470751224596704554150864598x_5^4\\
  &-195278103772230475236462201852070232583953616027553465198362842948766419805745572080063898831820910760883150x_5^3\\
  &+19077229110302732393589918306715874494051582708585559469272326336390718508460877526710123484137952808852276x_5^2\\
  &-1182613115373742448605260562696400057836772355234957546845962252089854311499956589035815582233191712586248x_5\\
  &+22822715865720062256822260809119245287191259217569053122320779895326763303740236022841002163082496x_1\\
  &+34970112057581607639480208646156766876639044746724958622577550768034487124422136034942339429080145084988.
\end{align*}
}

\bibliographystyle{elsarticle-harv} 
\bibliography{GEB-submission/edan}

\end{document}

%% End of file `elsarticle-template-harv.tex'.